\documentclass[lettersize,journal]{IEEEtran}
\def\BibTeX{{\rm B\kern-.05em{\sc i\kern-.025em b}\kern-.08em
    T\kern-.1667em\lower.7ex\hbox{E}\kern-.125emX}}
\usepackage{balance}
\usepackage{amssymb}
\usepackage{amsmath,amsfonts}
\usepackage{algorithmic}
\usepackage{algorithm}
\usepackage{array}
\usepackage{textcomp}
\usepackage{stfloats}
\usepackage{url}
\usepackage{graphicx}
\usepackage{cite}
\usepackage[caption=false]{subfig}
\usepackage[font=small,skip=2pt]{caption}
\usepackage{bm}
\usepackage{here}
\usepackage{amsthm}
\newtheorem{theorem}{Theorem}[section]

\newtheorem{definition}{Definition}

\usepackage{graphicx}
\usepackage{textcomp}
\usepackage{comment}
\usepackage{fancyhdr}

\fancypagestyle{plain}{
  \fancyhf{} 
  \fancyhead[C]{\footnotesize
© 2020 The authors. This manuscript is a revised author's preprint based on a paper originally presented at the Workshop on JAXA, ``Astrodynamics and Flight Mechanics,'' Sagamihara, Japan, C-6, 2020, and has not been peer-reviewed.
}
}

\begin{document}
\title{Time-Varying Kinematics Control for Magnetically-Actuated Satellite Swarm\\without Additional Actuator}
\author{Yuta Takahashi$^{1}$
\thanks{$^{1}$Graduate Student, Department of Mechanical Engineering, Tokyo Institute of Technology, Ookayama Meguro, Tokyo 152-8552, Japan{\tt\footnotesize takahashi.y.cl@m.titech.ac.jp}},
Hiraku Sakamoto$^{2}$
\thanks{$^{2}$Professor, Department of Mechanical Engineering, Tokyo Institute of Technology, Ookayama Meguro, Tokyo 152-8552, Japan},
Shin-ichiro Sakai$^{3}$
\thanks{$^{3}$Professor, Department of Spacecraft Engineering, Japan Aerospace Exploration Agency, 3-1-1 Yoshinodai, Sagamihara, Kanagawa 252-5210, Japan}
}
\maketitle
\begin{abstract}
Electromagnetic Formation Flight is a technology that uses electromagnetic forces and torques to control multiple satellites without conventional fuel-based propulsion. In this paper, the controllability of the system is discussed based on the conservation of the entire system's angular momentum, which constitutes a nonholonomic constraint. This paper designs a new controller for multiple satellites without an additional attitude actuator.
\end{abstract}
\begin{IEEEkeywords}
Electromagnetic Formation Flight, Spacecraft Swarm, Distributed Space System, Nonholonomic Mechanical System
\end{IEEEkeywords}
%
\thispagestyle{plain}  
\section*{Nomenclature}
\vbox{\noindent\setlength{\tabcolsep}{0mm}%
\begin{tabular}{p{25mm}cl} %
\hfil $L_{\bm f}V(\bm x)$\hfil & :\hspace{4mm} & Lie derivative, $L_{\bm f}V(\bm x)=\frac{\partial V(\bm x)}{\partial \bm x}\bm f \in \mathbb{R}$\\
\hfil $\bm\Delta$\hfil & :\hspace{4mm} & accessibility distribution
\end{tabular}}

\noindent{Subscripts}

\vbox{\noindent\setlength{\tabcolsep}{0mm}%
\begin{tabular}{p{25mm}cl}
 \hfil$AC$\hfil  & :\hspace{4mm} & alternating current\\
 \hfil$EMFF$\hfil  & :\hspace{4mm} & Electromagnetic Formation Flight\\
 \hfil$RWs$\hfil  & :\hspace{4mm} & 3-axis reaction wheels\\
  \hfil$STLC$\hfil  & :\hspace{4mm} & small-time local controllability
\end{tabular}}
\section{Introduction}
\label{sec1}
Maintaining distributed space systems comprising a large number of satellites without conventional fuel-based propulsion could enable more robust, functional, and advanced long-term observation missions. These missions include sparse-aperture sensing, stellar interferometry, distributed antenna arrays, and distributed space telescopes \cite{chung2011swarms,hadaegh2014development,marco2019distributed}, which are difficult to realize with monolithic satellites. To be cost-effective compared to monolithic satellites, the swarm satellite must be smaller. On the other hand, maintaining a relative distance often requires constant control, thereby increasing the initial fuel reserve and the satellite's size. Electromagnetic Formation Flight (EMFF) can control multiple satellites without consuming propellant \cite{sakai2008electromagnetic,kwon2011electromagnetic,porter2014demonstration,youngquist2013alternating,nurge2016satellite,sunny2019single,huang2014stability,kaneda2008relative,miller2010control,ahsun2006dynamics,elias2007electromagnetic,schweighart2006electromagnetic,buck2013path,ramirez2010new,ayyad2019optimal,kaneda2008relative,abbasi2019decentralized,abbasi2020decentralized,zhang2016angular,takahashi2021simultaneous,takahashi2022kinematics}. EMFF works on the principle that satellites can be controlled by electromagnetic forces and torques between the satellites that result from the interactions between coil-generated electromagnetic fields. However, in most previous studies, all satellites were assumed to have reaction wheels (RWs) in addition to electromagnetic coils. This requirement determines the lower limit of the satellite size of the EMFF system. Previous study \cite{takahashi2021simultaneous,takahashi2022kinematics} suggests that the number of RWs in the EMFF system can be significantly reduced by treating EMFF as a nonholonomic mechanical system. EMFF is a nonholonomic mechanical system because it can only output electromagnetic forces and torques that satisfy the conservation of angular momentum of the entire system, which imposes nonholonomic constraints. 

Most previous studies of EMFF assume that all satellites have RWs; otherwise EMFF requires complicated control scheme specific to the system because nonholonomic mechanical systems \cite{brockett1983asymptotic,coron1992global,pomet1992explicit,morin1999design,m1993nonholonomic,m1997exponential,m1998time,pettersen1996position,pomet1994exponential,rosier1992homogeneous,kawski1990homogeneous,morin2002time} become uncontrollable by linearization; all state can not be exponentially converged to the target state by smooth state feedback \cite{brockett1983asymptotic}, and any exponentially stabilizing solutions are necessarily non-Lipschitz \cite{m1997exponential}. In previous studies of EMFF \cite{takahashi2021simultaneous,takahashi2022kinematics}, the conservation of the system's angular momentum was formulated and incorporated into the control law. This study avoided nonholonomic properties and achieved smooth state-feedback control using RWs, unlike the purpose of this paper. Previous studies of nonholonomic mechanical systems have shown that periodic time-varying feedback can stabilize control-linear nonholonomic systems to target states without additional actuators \cite{coron1992global}. Constructive control design methods for control-linear nonholonomic systems have been proposed \cite{pomet1992explicit,morin1999design}. Of these, the homogeneous feedback method \cite{morin1999design,kawski1990homogeneous,m1993nonholonomic,m1997exponential,m1998time,morin2002time,pettersen1996position},,which ensuresg exponential convergence,, has been extended to control-affine nonholonomic systems. 
This method requires deriving a control law specific to the system that has not been derived in previous EMFF studies. 

Therefore, the objective of this research is to control the EMFF system without using RWs for the realization of distributed space systems of small satellites with the EMFF system. After controllability analysis of the Ethe MFF system without RWs, this paper derives a time-varying feedback control law, which guarantees that the origin of the Ethe MFF system is locally exponentially stable. The effectiveness of the designed control law is demonstrated through simulations of the formation reconfiguration for three satellites. 
\section{Preliminaries}
\label{sec2}
This section outlines basic formulations for the controller design of EMFF without using RWs. First, an averaging and modulation technique based on the AC method of EMFF \cite{ayyad2019optimal} is outlined. Then, the kinematics and dynamics of the EMFF, derived in previous studies \cite{takahashi2021simultaneous,takahashi2022kinematics}, are introduced.
\subsection{First Order Averaging of Dipole Modulation of Alternating Current Method}
\label{sec2-1}
This subsection outlines 
an averaging and modulation technique based on the AC method of EMFF \cite{ayyad2019optimal}. In particular, the AC modulation technique using sine and cosine waves is shown.

In the EMFF system, each satellite is equipped with 3-axis electromagnetic coils and 
controlled by adjusting the value of the ``dipole moment,'' which is proportional to the current $c$. For a circular coil, the relationship between the AC-driven dipole moments $\bm{\mu}_j(t)$ of the $j$-th satellite and the alternating current is expressed in Eq.~\eqref{sec2:eq13}:
\begin{equation}
  \label{sec2:eq13}
  \begin{aligned}
        \bm\mu_j(t)&=N_tA\bm nc(\sin(\omega_{f}t)+\cos(\omega_{f}t))\\
        &=\bm\mu_{j}^{\sin}\sin(\omega_{f}t)+\bm\mu_{j}^{\cos}\cos(\omega_{f}t)\\
  \end{aligned} 
\end{equation}
where $A$ is the area enclosed by the coil, $N_t$ is the number of coil turns, and $\bm n$ is the unit vector perpendicular to the plane of the coil. $\bm\mu_{(\sin)}$ and $\bm\mu_{(\cos)}$ are the AC amplitudes of the sine and cosine waves of the $j$-th satellite, respectively. Using the periodic dipole moment $\bm{\mu}_j(t)$, the time-varying electromagnetic force and torque of the AC method are derived. Because adapting the time-varying functions that control the system is complicated, they are approximated by first-order averaging over the period $T$ \cite{o1987averaging}. The averaged electromagnetic force $\bm f_{EM}^{(avg)}$ and torque $\bm\tau_{EM}^{(avg)}$ imparted by the system to the $j$-th dipole are expressed in Eq.~\eqref{sec2:eq15}. This approximation holds when the AC frequency $\omega_f$ is sufficiently high or when the dynamic frequency is sufficiently higher than the AC frequency \cite{o1987averaging}:
\begin{equation}
  \label{sec2:eq15}
  \left \{
  \begin{aligned}
        \bm f_{EMj}^{(avg)}&=\frac{1}{2}\sum_{j=1}^{n} \left(\bm{f}(\bm\mu_{i}^{\sin},\bm\mu_{j}^{\sin}, \bm r_{ij})+\bm{f}(\bm\mu_{i}^{\cos},\bm\mu_{j}^{\cos}, \bm r_{ij})\right)\\
        \bm{\tau}_{EMj}^{(avg)}&=\frac{1}{2}\sum_{j=1}^{n} \left(\bm{\tau}(\bm\mu_{i}^{\sin},\bm\mu_{j}^{\sin}, \bm r_{ij})+\bm{\tau}(\bm\mu_{i}^{\cos},\bm\mu_{j}^{\cos}, \bm r_{ij})\right)
  \end{aligned} 
  \right.
\end{equation}
where the electromagnetic force $\bm f (\bm\mu_i,\bm\mu_j, \bm{r}_{ij})$, and the electromagnetic torque $\bm\tau (\bm\mu_i,\bm\mu_j, \bm r_{ij})$ exerted on the $j$-th dipole by the $i$-th dipole are expressed as \cite{schweighart2006electromagnetic}
\begin{equation}
  \label{sec2:eq9}
  \begin{aligned}
    \bm{f}(\bm\mu_i,\bm\mu_j, \bm r_{ij})
    =&\frac{3\mu_0}{4\pi}\left(
            \frac{\bm\mu_i \cdot \bm\mu_j}{\|\bm{r}_{ij}\|^5} \bm r_{ij} \right .+
            \frac{\bm\mu_i \cdot \bm r_{ij}}{\|\bm r_{ij}\|^5}\bm\mu_{j}+
            \frac{\bm\mu_j \cdot \bm r_{ij}}{\|\bm r_{ij}\|^5}\bm\mu_{i}\\
             -&5\frac{(\bm\mu_i \cdot \bm r_{ij})(\bm\mu_j\cdot \bm r_{ij})}{\|\bm r_{ij}\|^7}\left . \bm r_{ij}\right),
   \end{aligned}
\end{equation}
\begin{equation}
  \label{sec2:eq10}
  \begin{aligned}
     \bm\tau (\bm\mu_i, \bm\mu_j, \bm r_{ij}) 
          &=\frac{\mu_0}{4\pi} \bm\mu_j \times \left(
        \frac{3\bm r_{ij}(\bm\mu_i \cdot \bm r_{ij})}{\|\bm r_{ij}\|^5}-\frac{\bm\mu_i}{\|\bm r_{ij}\|^3}
        \right),
\end{aligned}
\end{equation}
where $\bm r_{ij}$ is the position vector of the $j$-th satellite viewed from the $i$-th satellite. 
\subsection{Kinematics of EMFF without Attitude Actuator}
\label{sec2-2}
In this subsection, the kinematics of EMFF are derived using 
a nonholonomic constraint of the EMFF system as a preparation for the control design.
 
By using the conservation of linear momentum, $\bm r_1$ can be eliminated from the position vector $\bm r$ \cite{takahashi2021simultaneous,takahashi2022kinematics}, i.e., 
$
    ^ir=
    \begin{bmatrix}
   \ ^i{r}_2^{\mathrm{T}},
   \cdots,
   \ ^i{r}_{n}^{\mathrm{T}}
    \end{bmatrix}^{\mathrm{T}}\in \mathbb{R}^{(3n-3)\times 1}
$. Now, the kinematics of EMFF is derived using the system's nonholonomic constraints.
By assuming the angular momentum $\bm L$ to be 0, the conservation of $\bm L$ 
are expressed as 
\begin{equation}
   \label{sec3:eq2-3}
   \begin{aligned}
   \sum_{j=1}^n \left( m_i(\bm r_{j}-\bm r_{1})\times \frac{^Id}{dt}{\bm r}_j + \bm I_j \cdot \bm\omega_j\right)  = \bm 0\\
   \end{aligned}
\end{equation} 
Then, the inertial coordinate component of $\bm L$ can be expressed as
\begin{equation}
   \label{sec3:eq3}
   \begin{aligned}
       \sum_{j=1}^n  \left(m_j\left(^i\tilde{r}_j-^i\tilde{r}_1\right)\ ^i\dot{r}_j + C^{I/B_j}J_j\ ^{b_j}\omega_j\right) = 0 \\
    \end{aligned},
\end{equation}  
By using the EMFF states $\zeta$, the angular momentum of the system (see Eq.~\eqref{sec3:eq3}) is expressed as
\begin{equation}
    \label{sec3:eq5}
     \left \{
   \begin{aligned}
    A\zeta=&0\ ,\ \zeta=
    \left[
   ^i\dot{r}^{\mathrm{T}},\ 
   ^b\omega^{\mathrm{T}}
    \right]^{\mathrm{T}}\in \mathbb{R}^{(6n-3)\times 1}.\\
   A=&\left [
   m_2(^i\tilde{r}_2-\ ^i\tilde{r}_1)\right . ,\cdots,m_n(^i\tilde{r}_{n}-\ ^i\tilde{r}_1),\ \\
   &C^{I/B_1}J_1,\cdots,C^{I/B_{n}}J_n ]\in \mathbb{R}^{3\times(6n-3)}
    \end{aligned}
    \right.
  \end{equation}  
Let $S\in \mathbb{R}^{(6n-3)\times(6n-6)}$ be defined as a smooth and linearly independent vector field 
full rank matrix corresponding to the null space of matrix $A$, i.e.,
\begin{equation}
   \label{sec3:eq6}
   \left \{
   \begin{aligned}
  S=&NullSpace(A)=
 \begin{bmatrix}
      E_{(6n-6)} \\
      -C^{B_{n}/I}A_s
    \end{bmatrix}\in \mathbb{R}^{(6n-3)\times(6n-6)}\\
   A_s=&\left [
   m_2(^i\tilde{r}_2-\ ^i\tilde{r}_1)\right . ,\cdots,m_n(^i\tilde{r}_{n}-\ ^i\tilde{r}_1),\ \\
   &C^{I/B_1}J_1,\cdots,C^{I/B_{n-1}}J_{n-1} ]\in \mathbb{R}^{3\times(6n-6)}
    \end{aligned}
    \right.
  \end{equation}  
where $S$ denotes the tangent space of the manifold in which the angular momentum of the system does not change.  
By using the MRP kinematic differential equation of the $i$-th satellite \cite{schaub2003analytical}, the kinematics of the EMFF system can be expressed as
\begin{equation}
  \label{sec3:eq7}
  \left \{
  \begin{aligned}
    \dot{q}&=\hat{Z}\zeta\in \mathbb{R}^{(6n-3)\times 1}\\
    \zeta&=Sv\in \mathbb{R}^{(6n-3)\times 1}
  \end{aligned}
  \right .\ , \  
   \hat{Z}=
    \begin{bmatrix}
    E_{3n-3}&0\\
    0&[Z]
  \end{bmatrix}\in \mathbb{R}^{(6n-3)\times (6n-3)},
\end{equation}
where $
\dot{q}=
    \begin{bmatrix}
   \ ^i\dot{r}^{\mathrm{T}},
   \dot{\sigma}^{\mathrm{T}}
    \end{bmatrix}^{\mathrm{T}}\in \mathbb{R}^{(6n-3)\times 1}
$ \ is the time derivative of the generalized coordinates and $[Z]\in \mathbb{R}^{3n\times 3n}$ is defined by following the MRP kinematic differential equation of the system, which is expressed as
\begin{equation}
  \begin{aligned}
  \label{sec2:eq4}
   \dot\sigma&=
   \begin{bmatrix}
    \dot\sigma_1\\
    \vdots\\
    \dot\sigma_n\\
  \end{bmatrix}
   =[Z]\begin{bmatrix}
    ^{b_1}\omega_1\\
    \vdots\\
    ^{b_n}\omega_n, \\
  \end{bmatrix}\\
  &=
    \begin{bmatrix}
    Z(\sigma_1)&&0\\
    &\ddots&\\
    0&&Z(\sigma_n)\\
  \end{bmatrix}
  \begin{bmatrix}
    ^{b_1}\omega_1\\
    \vdots\\
    ^{b_n}\omega_n, \\
  \end{bmatrix}\in \mathbb{R}^{3n\times 1}\\
   Z(\sigma_i)&=\frac{1}{4}\left[(1-\sigma_i^{\mathrm{T}}\sigma_i)E_3+2\tilde{\sigma}_i+2\sigma_i\sigma_i^{\mathrm{T}}\right]\\
  \end{aligned}
\end{equation}
\subsection{Electromagnetic Formation Flight Dynamics}
\label{sec2-3}
This section outlines the dynamics of the EMFF system. The equation of motion describing the EMFF system is shown in Eq.~\eqref{sec3:eq8} ~ \cite{takahashi2021simultaneous,takahashi2022kinematics};  it combines the relative translational dynamics with the attitude dynamics of rigid spacecraft.
\begin{equation}
  \label{sec3:eq8}
  \ [M]\dot{\zeta}+
    [C]\zeta
  =
    u_c-A^{\mathrm{T}}\eta
\end{equation}
where
$
u_c=
\begin{bmatrix}
   \ ^if_c^{\mathrm{T}},
   \ ^{b}\tau_c^{\mathrm{T}}
    \end{bmatrix}^{\mathrm{T}}
    \in \mathbb{R}^{(6n-3)\times 1}$ is the control input, and $\eta\in \mathbb{R}^{3\times 1}$ is the vector of constraint forces.
 Matrices $[M]\in \mathbb{R}^{(6n-3)\times (6n-3)}$, $[C]\in \mathbb{R}^{(6n-3)\times (6n-3)}$ are expressed as 
\begin{equation}
  \label{sec3:eq9}
  \ [M]=
  \begin{bmatrix}
    [M_p]&0\\
    0&[M_a]\\
  \end{bmatrix}\,\ 
  [C]=
  \begin{bmatrix}
    0_{3n}&0\\
    0&[C_a]
  \end{bmatrix}\in \mathbb{R}^{(6n-3)\times (6n-3)}.
\end{equation}
Since the position of the 1st satellite $\bm r_1$ can be removed from the generalized coordinates \cite{takahashi2021simultaneous,takahashi2022kinematics}, mass matrix $[M_p]\in \mathbb{R}^{(3n-3)\times (3n-3)}$ of all satellites except the
1st satellite can be expressed as
  \begin{equation}
  \label{sec2:eq2}
  \ [M_p]=diag(diag(m_2,m_2,m_2),\cdots ,diag(m_n,m_n,m_n))
\end{equation}
Following Eq.~\eqref{sec2:eq5-2}, the matrices $[M_a]\in \mathbb{R}^{3n\times 3n}$ and $[C_a]\in \mathbb{R}^{3n\times 3n}$ are defined with respect to the attitude motion 
of the entire system:  
  \begin{equation}
  \label{sec2:eq5-2}
  \left \{
  \begin{aligned}
  \ [M_a]&=diag(J_1,\cdots,J_n)
 \,\ 
 \ 
 \\
 [C_a]&=diag\left(-(J_1\ ^{b_1}\omega_1)\ \bm{\tilde{}},\cdots ,-(J_n\ ^{b_n}\omega_n)\ \bm{\tilde{}}\right)
  \end{aligned}
  \right.
\end{equation}
Now, substituting $\zeta=Sv$ and $\dot{\zeta}=\dot{S}v+S\dot{v}$ 
into Eq.~\eqref{sec3:eq8}, and multiplying by the matrix $S^{\mathrm{T}}$ yields 

\begin{equation}
\label{sec3:eq10}
    \overline{M}\dot{v}+\overline{C}v
    =S^{\mathrm{T}}
    u_c,
\end{equation}
where $\overline{M}=S^{\mathrm{T}}[M]S\in \mathbb{R}^{(6n-6)\times (6n-6)}$ is a symmetric and positive definite matrix and
$\overline{C}=S^{\mathrm{T}}([M]\dot{S}+[C]S)\in \mathbb{R}^{(6n-6)\times (6n-6)}$.
In this case, the constraint force term $-A^{\mathrm{T}}\eta$
disappears, and the equation does not account for the constraint force due to angular momentum conservation.
\section{Nonholonomic Controller Design of Electromagnetic Formation Flight}
\label{sec3}
In this section, the controllability of the EMFF system is discussed using the conservation of angular momentum, which constitutes the system's nonholonomic constraints. This section designs a new controller for multiple EMFF satellites without an additional attitude actuator.
\subsection{Local Controllability of EMFF}
In this subsection, it is confirmed that the EMFF system of Eq.~\eqref{sec3:eq8} satisfies the small-time local controllability (STLC), mainly based on Sussmann's work \cite{sussmann1987general}. This controllability property guarantees the existence of a piecewise analytic feedback law \cite{sussmann1979subanalytic} and continuous time periodic feedback laws \cite{coron1995stabilization} that asymptotically steer arbitrary states into every state. STLC is defined as follows.

\textbf{Definition 1 \cite{sussmann1987general}.}\ \textit{
If some trajectory such that $x(0)= p$ reach $x(T)=q$, then $q$ will be said to be reachable from p in time T. The set of all q that are reachable from p in time T for the system is the time T reachable set from p, and will be denoted by $Reach(T, p)$. Similarly, $Reach(\leq T, p)$ is expressed by the following equation.  
\begin{equation}
  Reach(\leq T, p)=\bigcup_{0\leq t\leq T} Reach(t, p)
\end{equation}
The system is STLC from $p$ if $p$ is an interior point of $Reach(\leq T, p)$ for all $T> 0$.}

First, EMFF dynamics and kinematics of Eqs.~\eqref{sec3:eq8} and \eqref{sec2:eq4} are modified to the standard control system form for discussion of the STLC property of EMFF. Let the control input $u_c$ be defined as $u_c=[M]Su$ with a new input $ u\in\mathbb{R}^{(6n-6)}$. The equation of motion of Eq.~\eqref{sec3:eq8} is expressed by the following Eq.~\eqref{eq:4-1}.

\begin{equation}
  \label{eq:4-1}
   \begin{aligned}
  &\left \{
 \begin{aligned}
  \dot{q}&=\hat{Z}\zeta\\
 \dot{\zeta}&=-\ [M]^{-1}
    [C]\zeta
  +\sum_{i=1}^{6n-3} S_i(x)u_i
 \end{aligned}
 \right .\\
 &S=\begin{bmatrix}
   S_1\in \mathbb{R}^{(6n-3)}&\cdots&S_{6n-6}\in \mathbb{R}^{(6n-3)}
 \end{bmatrix}
\end{aligned}
\end{equation}
Now, Eq.~\eqref{eq:4-1} defines EMFF state $x=[q^{\mathrm{T}},\zeta^{\mathrm{T}} ]^{\mathrm{T}}$, a drift vector field $f_0$, and control vector fields  $g_i(i=1,\ldots,6n-6$), according to the following standard control system form of Eq.~\eqref{eq:3-2}.
\begin{equation}
  \label{eq:3-2}
  \begin{aligned}
\dot{x}&=f_0(x)+\sum_{i=1}^{6n-6} g_i(x)u_i\\
f_0(x)&=\begin{bmatrix}
  \hat{Z}\zeta\\
  -\ [M]^{-1}[C]\zeta
\end{bmatrix},\ g_i(x)=\begin{bmatrix}
   0_{6n-3}\\S_{i}
\end{bmatrix}\in \mathbb{R}^{(12n-6)}
  \end{aligned}
\end{equation}
where $x\in \mathbb{R}^{(12n-6)\times 1}$, $u\in \mathbb{R}^{(6n-6)\times 1}$.

Next, local accessibility, which is a necessary condition for STLC, is examined. This condition is also called Lie Algebra Rank Condition \cite{sussmann1987general}. The definition of local accessibility is given by the Lie bracket and the accessibility distribution $\bm\Delta$, which represents the open set reachable in the neighborhood of $\bm x$.
\begin{definition}[\cite{nijmeijer1990nonlinear}]\ \textit{The system of Eq.~\eqref{eq:3-2} is locally accessible from $x_0$ if $Reach(T,x_0)$ contains a non-empty open set for all neighborhoods $V$ of $x_0$ and all $T> 0$. If this holds for any $x_0$, then the system is called locally accessible.}
\end{definition}
\begin{definition}
The Lie bracket of $\bm f$ and $\bm g$ denoted by [$\bm f$,$\bm g$] is a third vector field defined by
\begin{equation}
  \begin{aligned}
     \ [f,g](q)= \frac{\partial g}{\partial q}(q)f(q)-\frac{\partial f}{\partial q}(q)g(q)   
\end{aligned}
\end{equation}
\end{definition}
\begin{definition}
Consider all the vector fields obtained from the following Lie Bracket of $\bm f_i$.
\begin{equation}
  \begin{aligned}
    \ [\bm r_k,[\bm r_{k-1},[\cdots,[\bm r_2,\bm r_1]\cdots]]]\ \ k=2,3,4,\cdots\\
    \bm r_i \in \{\bm f_0, \bm g_i\cdots, \bm g_{6n-6} \}
  \end{aligned}
\end{equation}
The linear space $\bm\Delta$ in which these vector fields and the vector field $\bm f_i$ extend is called the accessibility distribution.
\end{definition}
With the above definitions, the theorem about local accessibility is shown as follows.
\begin{theorem}[\cite{nijmeijer1990nonlinear,sussmann1987general}]The system of Eq.~\eqref{eq:3-2} is local accessibility if dim($\bm\Delta$) = $12n-6$ at every $x$
\end{theorem}
Then, a simple calculation using the Philip Hall basis \cite{murray2017mathematical} of Eq.~\eqref{eq:3-2} (see Sec.~\ref{Appendix}) showed that the system satisfied local accessibility at every $x$ from the following results.
\begin{equation}
  \label{eq:3-3-2}
  \begin{aligned}
\bm\Delta=span\left\{[f_0,g_i],[g_1,[f_0,g_2]],[g_2,[f_0,g_3]],[g_3,[f_0,g_1]],\right .\\
\left .g_i,[g_1,g_2],[g_2,g_3],[g_3,g_1]\right\}=\mathbb{R}^{12n-6},\ (i=1,\cdots,6n-6)
  \end{aligned}
\end{equation}
Note that local accessibility is only a necessary condition for drift systems to satisfy STLC, although STLC follows from local accessibility for driftless systems \cite{nijmeijer1990nonlinear,hermann1977nonlinear}. 

Before confirming if the system of Eq.~\eqref{eq:3-2} meets the sufficient conditions of STLC, the classification of Lie bracket is introduced. For an arbitary Lie bracket $h$, $\delta_0(h)$ and $\delta_i(h)$ indicate the number of times that a drift vector field $f_0$ and control vector fields $g_i(i=1,\ldots,6n-6)$ appear in $h$, respectively. At this time, if $\delta_0(h)$ is odd and all $\delta_i(h)$ are even, the Lie bracket $h$ is defined as \textit{bad}. Otherwise, the bracket $h$ is defined \textit{good}. Using this classification, the sufficient conditions for STLC are given by the following theorem.
\begin{theorem}[\cite{sussmann1987general}]The system is STLC at zero velocity if every bad symmetric product is a linear combination of lower-order good symmetric products.
\end{theorem}

By Eq.~\eqref{eq:3-3-2}, it is shown that the system of Eq.~\eqref{eq:3-2} satisfies the sufficient conditions of STLC.
These controllability analysis guarantees
\subsection{Time-Varying Kinematics Controller Design}
This subsection derives a time-varying feedback control law that guarantees that the origin of the EMFF system is locally exponentially stable, based on a previous study \cite{m1998time}. By combining EMFF kinematics of Eq.~\eqref{sec3:eq7} and EMFF dynamics of Eq.~\eqref{sec3:eq10}, the following dynamical model is derived.
\begin{equation}
  \label{eq:4-2-2}
    \left \{
  \begin{aligned}
    \dot{q}&=\hat{Z}Sv\\  
    \overline{M}\dot{v}&=-\overline{C}(q,\dot{q})v+S^{\mathrm{T}}u_c
\end{aligned}
    \right.
\end{equation}
Given $n$ satellites with the dynamical model described by Eq.~\eqref{eq:4-2-2}. Then, the control law $u_c$ in Eq.~\eqref{eq:4-2-3} is applied. Then, the origin of the system is  guaranteed locally exponentially stable.
\begin{equation}
  \label{eq:4-2-3}
   \begin{aligned}
    u_c&=
\begin{bmatrix}
   \ ^if_{EM}^{(avg)\mathrm{T}},
   \ ^{b}\tau_{EM}^{(avg)\mathrm{T}}
    \end{bmatrix}^{\mathrm{T}}=[M]Su\\
    &\left \{
  \begin{aligned}
    u&=-K(v-w)\\
   w&=\sum_{j=1}^{6n-6} S_{j} v_j^1
            +\sum_{j=6n-5}^{6n-3} u_j^1\hat{S}_{j}^1 v_j^1
            +\sum_{j=6n-5}^{6n-3} u_j^2\hat{S}_{j}^2 v_j^2\\
   \end{aligned}
   \right .
   \end{aligned}
 \end{equation}

\begin{equation}
   \label{eq:4-2-3-2}
  \begin{aligned}
    &\left \{
      \begin{aligned}
    v_i^1&=  
    \left \{
   \begin{aligned}
   &\tilde{u}_i\ &  \\
   &\tilde{u}_i+h_i^2 \ &(i = 3n-5, \ldots, 3n-3)  \\
   &\rho \ &(i=6n-5, \ldots, 6n-3)  \\
  \end{aligned}
    \right. \\
   v_i^2&=  
   \begin{aligned}
   &\frac{\tilde{u}_i}{\rho} \ &(i=6n-5, \ldots, 6n-3)\\
  \end{aligned}\\
  \end{aligned}
  \right .\\
  &\left \{ 
      \begin{aligned}
    u_{6n-5}^1&=\epsilon^{-\frac{1}{2}}\cos(\omega_1 t/\epsilon),\ u_{6n-5}^2=2\omega_1\epsilon^{-\frac{1}{2}}\sin(\omega_1 t/\epsilon)\\
    u_{6n-4}^1&=\epsilon^{-\frac{1}{2}}\cos(\omega_2 t/\epsilon),\ u_{6n-4}^2=2\omega_2\epsilon^{-\frac{1}{2}}\sin(\omega_2 t/\epsilon)\\
    u_{6n-3}^1&=\epsilon^{-\frac{1}{2}}\cos(\omega_3 t/\epsilon),\ u_{6n-3}^2=2\omega_3\epsilon^{-\frac{1}{2}}\sin(\omega_3 t/\epsilon)\\
    \end{aligned}
    \right .\\
    &\left \{
      \begin{aligned}
    \hat{S}_{6n-5}^1&=S_2,\ \hat{S}_{6n-5}^2=S_3\\
    \hat{S}_{6n-4}^1&=S_3,\ \hat{S}_{6n-4}^2=S_1\\
    \hat{S}_{6n-3}^1&=S_1,\ \hat{S}_{6n-3}^2=S_2\\
    \end{aligned}
    \right .\\
\end{aligned}
  \end{equation}
    \begin{equation}
      \begin{aligned}
   \tilde{u}&=-
   \begin{bmatrix}
    S&[S_2,S_3]&[S_3,S_1]&[S_1,S_2]
   \end{bmatrix}^{-1}
    \begin{bmatrix}
   ^i{r}^{\mathrm{T}},
   {\sigma}^{\mathrm{T}}
    \end{bmatrix}^{\mathrm{T}}\\
     &\left \{
   \begin{aligned}
  h_{3n-5}^2=v_{6n-3}^2(L_{S_{3n-4}}v_{6n-3}^1)-v_{6n-4}^1(L_{S_{3n-3}}v_{6n-4}^1)\\
  h_{3n-4}^2=v_{6n-5}^2(L_{S_{3n-3}}v_{6n-5}^1)-v_{6n-3}^1(L_{S_{3n-5}}v_{6n-3}^1)\\
  h_{3n-3}^2=v_{6n-4}^2(L_{S_{3n-5}}v_{6n-4}^1)-v_{6n-5}^1(L_{S_{3n-4}}v_{6n-5}^1)
  \end{aligned}
  \right .\\
    \rho&=\left(\sum_{j=2}^{n} \left(r_{jx}^4+r_{jy}^4+r_{jz}^4\right) + \sum_{j=1}^{n} \left(\sigma_{jx}^2+\sigma_{iy}^2+\sigma_{jz}^2\right)\right)\\
   \end{aligned}
 \end{equation}
where $K$ are positive scalar gains. 
\begin{theorem}
The control force given by the control law $u_c$ in Eq.~\eqref{eq:4-2-3} will not change the angular momentum of the system $\bm L$. 
\end{theorem}
\begin{proof}
Let $R(q)\in\mathbb{R}^{3\times(6n-3)}$ be defined as matrix that represents the following matrix: 
\begin{equation}
  \label{sec3:eq11-4}
  \begin{aligned}
 R=&\left [
   (^i\tilde{r}_2-\ ^i\tilde{r}_1)\right. ,\cdots,(^i\tilde{r}_{n}-\ ^i\tilde{r}_1),
   \ C^{I/B_1},\cdots,C^{I/B_{n}}]. 
  \end{aligned}
  \end{equation}
Matrix $R$ holds the following relationship of Eq.~\eqref{sec3:eq11-5}:
 \begin{equation}
  \label{sec3:eq11-5}
    R[M]S=AS=0 \ (\because S=NullSpace(A)).
\end{equation}
Based on Eq.~\eqref{sec3:eq11-5}, 
The rate of change of the angular momentum of the system $\ ^i\dot{L}$ caused by the control law is always 0, as shown by Eq.~\eqref{sec3:eq11-6}:
  \begin{equation}
  \label{sec3:eq11-6}
    Ru_c=R
\begin{bmatrix}
   \ ^if_c^{\mathrm{T}},
   \ ^{b}\tau_c^{\mathrm{T}}
    \end{bmatrix}^{\mathrm{T}}=\ ^i\dot{L}=0.
\end{equation}
\end{proof}
\section{Numerical Calculation}
\label{sec4}
To show the effectiveness of the designed control law, 
Simulations of formation reconfiguration in a three-satellite system controlled only by magnetic actuators, without using RWs, are conducted. Note that, in this numerical calculation, it is assumed that the AC frequency $\omega_f$ of Eq.~\eqref{sec2:eq13} is sufficiently large, and that the electromagnetic forces and electromagnetic torques could be completely approximated by the averaged values of Eq.~\eqref{sec2:eq15}. A mass of each satellite is set to be 3 kg and the inertia of the each satellite in the body-fixed frame are expressed as follows: 1-th satellite $J_{S/C1} = diag([1,2,3])$ kg$\cdot$m$^2$; 2-th satellite $J_{S/C2} = 2diag([1,2,3])$ kg$\cdot$m$^2$; 3-th satellite $J_{S/C3} = 3diag([1,2,3])$ kg$\cdot$m$^2$. It is assumed that satellite control is performed in an environment with no external forces. The initial values of the relative position and absolute attitude are random, and each target value is set to 0. Each gains $\epsilon$ and $K$ are set to be 0.1 and $30I_{12\times12}$. In addition, the angular frequencies of the time-varying terms in Eq.~\eqref{eq:4-2-3-2} are set as follows: $\omega_1=0.2$rad/s, $\omega_2=0.4$rad/s, and $\omega_3=0.6$rad/s. 

The results of controlling relative positions $^ir$ and absolute attitudes $\sigma$ via the designed control law are shown in Figs.~1 and 2, respectively. The plots of the three-dimensional position are also demonstrated in Fig.~3. It can be seen that the satellite system converges to the origin and the desired attitude of each satellite.

\begin{figure}[hbt!]
    \label{sec4:fig1}
\centering
\includegraphics[width=9truecm]{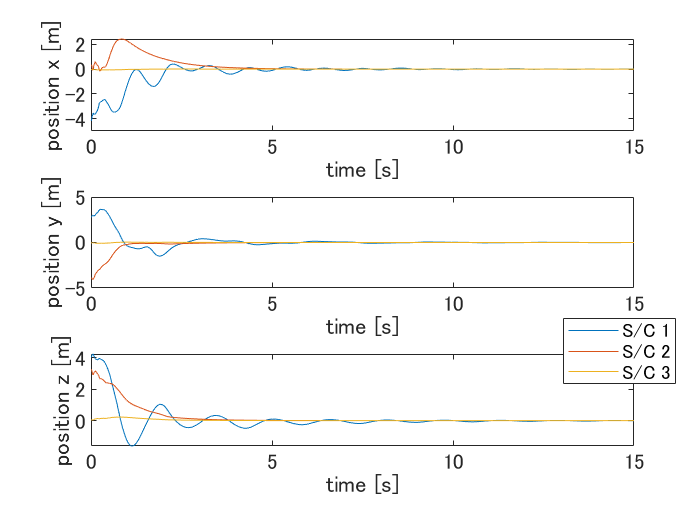}
\caption{Results of the simulation for three spacecraft-relative position-}
\end{figure}
\begin{figure}[hbt!]
    \label{sec4:fig2}
\centering
\includegraphics[width=9truecm]{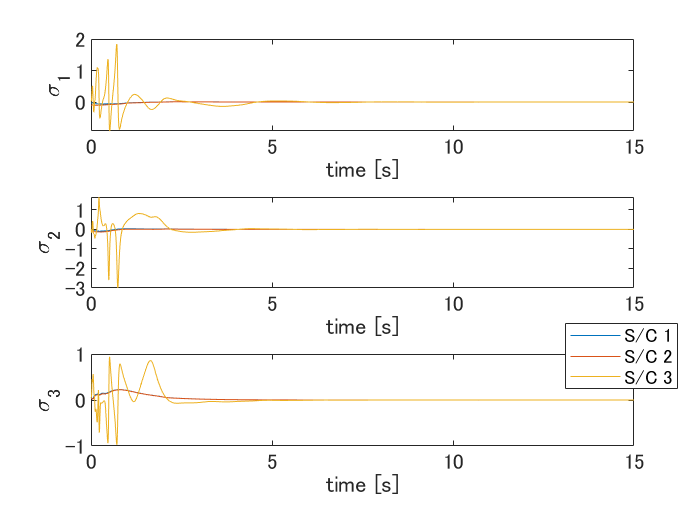}
\caption{Results of the simulation for three spacecraft-absolute attitude-}
\end{figure}
\begin{figure}[hbt!]
    \label{sec4:fig3}
\centering
\includegraphics[width=9truecm]{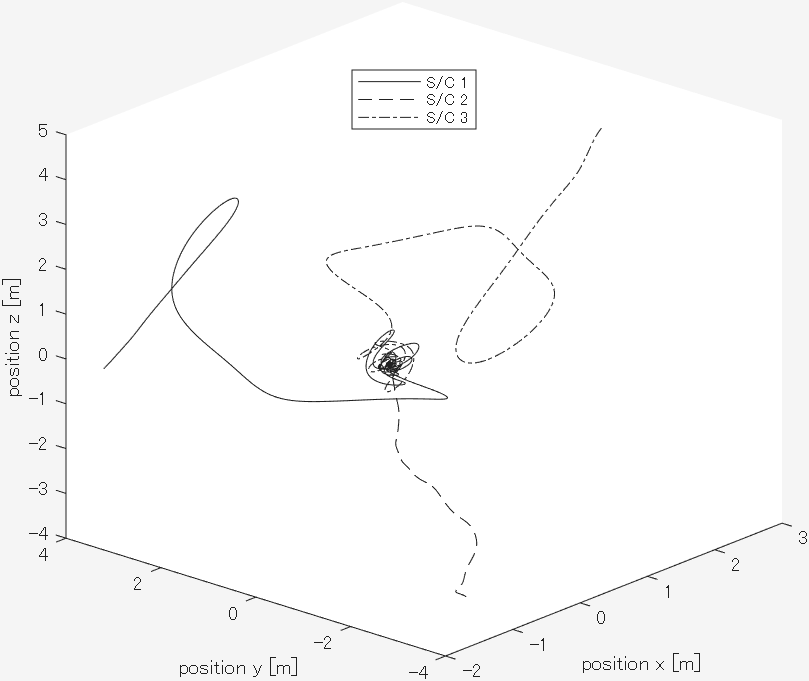}
\caption{Results of the simulation for three spacecraft-three-dimensional position-}
\end{figure}
\section{Conclusion}
\label{sec6}
This paper derives a time-varying feedback control law for the Electromagnetic Formation Flight system without using 3-axis reaction wheels, enabling the realization of distributed space systems consisting of small satellites. This controller guarantees that the system's origin is locally exponentially stable. After the controllability analysis of the Electromagnetic Formation Flight system without 3-axis reaction wheels, a new controller is designed using the conservation of angular momentum, a nonholonomic constraint of the system. Then, the effectiveness of the designed control law is demonstrated through simulations of formation reconfiguration for a three-satellite system. 
\section*{Appendix}
\label{Appendix}
\subsection*{Philip Hall Basis \cite{murray2017mathematical} of Eq.~\eqref{eq:3-2}}
\begin{equation}
  \begin{aligned}
     &[f_0,g_i]=\frac{\partial g_i}{\partial x}f_0-\frac{\partial f_0}{\partial x}g_i\\
     &=\begin{bmatrix}
       0_{(12n-9)\times (12n-6)}\\
       \begin{bmatrix}
       \frac{\partial S_i(6n-5)}{\partial x}\\
       \frac{\partial S_i(6n-4)}{\partial x}\\
       \frac{\partial S_i(6n-3)}{\partial x}
       \end{bmatrix}f_0
     \end{bmatrix}
     -\begin{bmatrix}
       E_{3n-3}&0_{(3n-3)\times 3n}\\
       0_{3n\times (3n-3)}&[Z]\\
       0_{(3n-3)\times (3n-3)}&0_{(3n-3)\times 3n}\\
       0_{3n\times (3n-3)}&\frac{\partial ([M]^{-1}[C]\zeta)}{\partial \omega}
     \end{bmatrix}S_i\\
\end{aligned}
\end{equation}
where $g_i=[0_{6n-3};S_{i}]$ and 
\begin{equation}
[S_1,S_2]=
\begin{bmatrix}
0_{(6n-6)\times 1} \\
\begin{bmatrix}
       \frac{\partial S_2(6n-5)}{\partial x}\\
       \frac{\partial S_2(6n-4)}{\partial x}\\
       \frac{\partial S_2(6n-3)}{\partial x}
\end{bmatrix}S_1
-
\begin{bmatrix}
       \frac{\partial S_1(6n-5)}{\partial x}\\
       \frac{\partial S_1(6n-4)}{\partial x}\\
       \frac{\partial S_1(6n-3)}{\partial x}
\end{bmatrix}S_2
\end{bmatrix}
\end{equation}
\begin{equation}
[S_2,S_3]=
\begin{bmatrix}
0_{(6n-6)\times 1} \\
\begin{bmatrix}
       \frac{\partial S_3(6n-5)}{\partial x}\\
       \frac{\partial S_3(6n-4)}{\partial x}\\
       \frac{\partial S_3(6n-3)}{\partial x}
\end{bmatrix}S_2
-
\begin{bmatrix}
       \frac{\partial S_2(6n-5)}{\partial x}\\
       \frac{\partial S_2(6n-4)}{\partial x}\\
       \frac{\partial S_2(6n-3)}{\partial x}
\end{bmatrix}S_3
\end{bmatrix}
\end{equation}
\begin{equation}
[S_3,S_1]=
\begin{bmatrix}
0_{(6n-6)\times 1} \\
\begin{bmatrix}
       \frac{\partial S_1(6n-5)}{\partial x}\\
       \frac{\partial S_1(6n-4)}{\partial x}\\
       \frac{\partial S_1(6n-3)}{\partial x}
\end{bmatrix}S_3
-
\begin{bmatrix}
       \frac{\partial S_3(6n-5)}{\partial x}\\
       \frac{\partial S_3(6n-4)}{\partial x}\\
       \frac{\partial S_3(6n-3)}{\partial x}
\end{bmatrix}S_1
\end{bmatrix}
\end{equation}
\bibliographystyle{IEEEtran}
\bibliography{references}
\end{document}